\documentclass[12pt]{amsart}
\usepackage{amsmath,amsfonts,latexsym,graphicx,amssymb,url}
\usepackage{amsmath,txfonts,pifont,bbding,pxfonts,manfnt}
\usepackage{psfrag}
\usepackage{empheq}
\usepackage{color}
\usepackage{mathrsfs}
\usepackage{wrapfig, subfigure,graphicx}
\usepackage{hyperref}
\usepackage[active]{srcltx}
\usepackage{breqn}
\usepackage{pdfsync}
\usepackage{pdflscape}
\def\Arg{\mathop{\operator@font Arg}\nolimits}
\theoremstyle{plain}
\newtheorem{theorem}{Theorem}[section]

\newtheorem{definition}[theorem]{Definition}
\newtheorem{remark}[theorem]{Remark}

\begin{document}

\setcounter{equation}{0}

\title[On Propagators in TDDFT]
{Existence of Propagators for time dependent Coulomb-like potentials}
\author[E. Stachura]
{Eric Stachura}

\address{Department of Mathematics and Statistics, Haverford College, 370 Lancaster Ave., Haverford, PA 19041}
\email{estachura@haverford.edu}

\maketitle
\begin{abstract}
We prove existence of propagators for a time dependent Schr\"odinger equation with a new class of softened Coulomb potentials, which we allow to be time dependent, in the context of time dependent density functional theory. We compute explicitly the Fourier transform of these new potentials, and provide an alternative proof for the Fourier transform of the Coulomb potential using distribution theory. Finally we show the new potentials are dilatation analytic, and so the spectrum of the corresponding Hamiltonian can be fully characterized. 
\end{abstract}
\tableofcontents

\section{Introduction}
 The classical Coulomb potential $V(x) = |x|^{-1}$ for $x \in \mathbb{R}^3$ is one of the most fundamental potentials in electronic structure theory since it describes the interaction between electron and proton for the Hydrogen atom. It can easily be extended to many electron systems by considering the electron repulsion term
 $$\sum_{j<k} \dfrac{1}{|x_j -x_k|}.$$ The difficulty surrounding this potential in practice is its singular behavior at the origin. In the past, the singularity has been smoothed out by multiplying by a $C^\infty$ radial function. One can then recover the classical Coulomb potential as a limiting case of the smoothed out potentials. Motivated by the Runge-Gross Theorem \cite{Runge-Gross} in Time Dependent Density Functional Theory (TDDFT) where singular time dependent potentials naturally arise, we study a new class of smoothed out Coulomb potentials that ``behave well" at the origin. Namely, these potentials are $C^\infty$ radial functions whose derivatives of every order vanish at the origin. Such a property should be satisfied by a potential if one hopes to prove a Runge-Gross type theorem (see Section \ref{sec:background}), at least via the existing proof techniques. 
 
  

We employ time dependent perturbation theory to obtain existence of propagators. This is not a new method; rather, its mathematical foundations can be found in the seminal work \cite{Simon73}, and has numerous applications (such as computing radiative lifetimes of excited states of various atoms). By using this technique we allow the nuclei to move in time, which is where the explicit time dependence for us arises. Additionally, time dependent perturbation theory arises in density functional theory \cite[Chapter 10]{marques2006time}, for instance to derive better approximations of certain exchange correlation functionals; thus there is a natural connection between the two. A nice review of the fundamentals of TDDFT can be found in \cite{maitra2016perspective}. 

Furthermore, the propagator we construct to prove existence of solutions to an appropriate Schr\"odinger equation is done by using the so-called Dyson series. We find this quite practical, as this series can be used to calculate various transition probabilities. Additionally, our propagator is constructed in such a way so that the wavefunction stays in the domain of the quantum harmonic oscillator for all positive time. 

A number of recent results address singular potentials, both time independent and time dependent, but their singular nature often causes tremendous difficulty. For instance, propagators are constructed in \cite{Yajima2016} for a time dependent Schr\"odinger equation via Strichartz estimates, but with particles subject to external fields. These results unfortunately do not include the time dependent potentials studied here. It is our hope that through the potentials introduced in this article, more rigorous analysis (spectral analysis of the time dependent Hamiltonian, regularity and smoothing properties of the propagators, etc.) can be done in the context of TDDFT, with an eye of improving the classical Runge-Gross Theorem. 

The outline of this article is as follows. In Section \ref{sec:background} we give the necessary background and introduce the new class of potentials in (\ref{Ppotential}). Then in Section \ref{sec:existence} we state and prove the main theorem for the Hydrogen atom. This is extended in a natural way to a $N$ particle, $M$ nuclei setting in Section \ref{sec:Nparticleexistence}. In Section \ref{sec:FA} we compute the Fourier transform of the new potential, and show that in an appropriate limit one recovers the Fourier transform of the classical Coulomb potential. Finally, in Section \ref{sec:spectral} we show that the new class of potentials introduced here are dilatation analytic, and so the spectrum of the corresponding Hamiltonian can be completely characterized. 

The author would like to thank John Perdew, Maxim Gilula, and Jonas Lampart for many helpful discussions related to this article.

\section{Background and Setup} \label{sec:background}
For a general non-relativistic $N$ particle system, the measure $dx = dx_1\cdots dx_N$ will always denote Lebesgue measure in $\mathbb{R}^{3N}$ since we take $x_j \in \mathbb{R}^3$ for $1 \leq j \leq N$. We frequently write $x = (x_1, ..., x_n)$ as a general element of $\mathbb{R}^n$. Recall that for $1\leq p < \infty$, 
$$L^p\left( \mathbb{R}^{3N}\right) = \left\{ \psi: \left(\int_{\mathbb{R}^{3N}} |\psi(x_1,..,x_N)|^p dx\right)^{1/p} < + \infty\right\}.$$
The space of square integrable functions on $\mathbb{R}^{3N}$ is $L^2 (\mathbb{R}^{3N})$. 

If $1 \leq p, q \leq \infty$, a function $f$ is said to belong to $L^p (\mathbb{R}^{3N}) + L^q(\mathbb{R}^{3N})$ if there exists $f_1 \in L^p (\mathbb{R}^{3N})$ and $f_2 \in L^q (\mathbb{R}^{3N})$ such that $f = f_1 +f_2$. Recall that $L^\infty(\Omega)$ denotes the set of functions on some space $\Omega$ which are bounded almost everywhere\footnote{For us this will always be with respect to Lebesgue measure.}. We denote by $B_r(x)$ the ball centered at $x$ of radius $r$ in $\mathbb{R}^n$, i.e. $B_r(x) = \{ y \in \mathbb{R}^n: |y-x| < r \}$. 
A function $V$ is said to belong to the class $C^k$ if all $k$ space derivatives of $V$ exist and are continuous. Given a function $V \in L^2 (\mathbb{R}^{3N})$, any derivative of $V$ must be understood in the sense of distributions (see Section \ref{sec:Distributions} for more details). 

We will also need the notion of the domain of an operator. To this end, suppose $H$ is a Hilbert space. An unbounded operator $A$ is an operator defined on a subset $\mathscr{D}(A) \subset H$ called its domain. In general $\mathscr{D}(A) \neq H$. In fact, it is well known that $\mathscr{D}(A) = H$ if and only if $A$ is bounded. In the quantum mechanical setting that follows, the operator $A$ will generally be a Hamiltonian, which most of the time is unbounded, so its domain needs to be specified. The adjoint of an unbounded operator is well-defined, provided its domain is dense in the Hilbert space. In particular, $A: \mathscr{D}(A) \subset H \to H$ is said to be self-adjoint provided $\mathscr{D}(A) = \mathscr{D}(A^*)$, where $A^*$ is the adjoint of $A$. Further, $[A, B]$ denotes the commutator of two operators $A,B$ whenever this makes sense: $[A,B]\psi = A(B\psi)-B(A\psi)$. 


Finally, recall that a bounded, linear operator $U(t)$ on a Hilbert space $H$ is said to be unitary provided $U(t) U(t)^* = U(t)^* U(t)= I$, where $U(t)^*$ denotes the adjoint of $U(t)$, and $I$ denotes the identity on $H$. 

\subsection{Density Functional Theory}

Consider a general $N$ particle (Fermion or Boson) system, with $M$ nuclei, and with Hamiltonian\footnote{We use atomic units, so that $e^2 = \hbar = m =1$ so distances are measured in Bohrs and energies in Hartrees.}
\begin{align} \label{Nparticletdhamiltoniandef}
H_V(t) \coloneqq \sum_{j=1}^N -\Delta_{x_j} + V_0(x_j) + V(x_j, t) + \sum_{j<k} F(x_j-x_k),
\end{align}
where 
$$\Delta_{x_j} \coloneqq \dfrac{\partial^2}{\partial x_{j_1}^2} +  \dfrac{\partial^2}{\partial x_{j_2}^2} + \dfrac{\partial^2}{\partial x_{j_3}^2}$$ is the Laplacian (kinetic energy operator) acting on particle $j$, and $V_0(x_j)$ is a fixed external potential. The function $F$ is the interaction term between electrons $j$ and $k$. Corresponding to $H_V(t)$ is the time dependent Schr\"odinger equation
\begin{align} \label{Nparticletdse}
\begin{cases}
\dfrac{\partial \psi (X,t)}{\partial t} &= -i H_V \psi(X,t)\\
\psi(X,0) &=\psi_0(X).
\end{cases}
\end{align}
Here, $X = (x_1, ..., x_N)$ and each $x_j \in \mathbb{R}^3$. To emphasize our focus on singular potentials, we shall take the fixed external potential to be Coulomb, namely,
\begin{align} \label{V0singular}
V_0(x_j) = -\dfrac{1}{|x_j|}.
\end{align}
Without the time dependent external potential we write
\begin{align} \label{H0def}
H_0 = \sum_{j=1}^N -\Delta_{x_j} +V_0(x_j) = \sum_{j=1}^N -\Delta_{x_j} -\dfrac{1}{|x_j|} + \sum_{j<k} F(x_j-x_k) 
\end{align}
so that the full time dependent Hamiltonian $H_V(t)$ can be written as a time dependent perturbation of $H_0$: 
\begin{align}\label{HVdecomp}
H_V(t) = H_0+ V(t),
\end{align}
where $V(t)$ denotes multiplication by $V(x_j,t)$. We focus here on the Hydrogen atom: when $N=M =1$ and $F \equiv 0$. In this case the free Hamiltonian is
\begin{align} \label{fullHamiltonianH0}
H_0 = - \Delta -\dfrac{1}{|x|}
\end{align}
and the full Hamiltonian is
\begin{align} \label{fullHamiltonian}
H_V = H_0 + V(x,t)
\end{align}
for some external, time dependent, potential. By the Kato-Rellich Theorem\footnote{See e.g. Theorem X.15 in \cite{RS2}.}, we have that $H_0$ is self-adjoint on $\mathscr{D}(-\Delta)$. The density $n(x,t)$ of the $N$ particle solution of (\ref{Nparticletdse}) is defined\footnote{We are ignoring spin coordinates here.} by
$$n(x,t) = N \int_{\mathbb{R}^{3(N-1)}} |\psi(x,x_2, ..., x_N, t)|^2 dx_2 ... dx_N.$$
We recall here the classical Runge-Gross uniqueness theorem, which serves as motivation:

\begin{theorem}[\cite{Runge-Gross}] \label{thm:RG}
Let $V_j(x,t)$ for $j=1,2$ be two time analytic potentials, and let $\psi_j(X,t)$ denote the corresponding solutions of (\ref{Nparticletdse}) with the same initial data $\psi_0$. Suppose the $\psi_j$ are also time analytic. If $n_1(x,t) = n_2(x,t)$ for all $t \in [0, t_{\text{max}})$ and all $x \in \mathbb{R}^3$, then
$$V_1(x,t) = V_2(x,t) + C(t).$$
\end{theorem}

A current issue with Theorem \ref{thm:RG} is the applicability to Coulombic systems, which are of interest in practice. For this reason, we undertake the task of smoothing out a Coulomb potential by introducing a decaying exponential term of the form $e^{-c/|x|}$ for some constant c. This is done in such a way so as to obtain the classical Coulomb potential in the limit as the parameter $c$ goes to zero. However, in order to begin such an analysis, one needs to prove existence of continuously differentiable solutions to (\ref{Nparticletdse}) with such softened potentials. The first goal of this article is to construct such solutions via a convergent Neumann series known as the Dyson series.

We treat the nucleus as a classical particle, and allow it to move along a trajectory described by $r(t)$. To this end, we smooth out the classical Coulomb potential and introduce the potential
\begin{align} \label{Ppotential}
V_P(x) = \dfrac{e^{-C/|x|}}{|x|}
\end{align}
for $x \in \mathbb{R}^3$ and $C>0$ a small (arbitrary) constant.

There is a subtle difference between the potentials $V_P(x)$ and other previously studied softened potentials, such as the Yukawa potentials $V_Y(x)$ \cite{Yukawa35} defined by
$$V_Y(x) = \alpha^2 \dfrac{e^{-c |x|}}{|x|},$$
which can be seen by analyzing their behavior at the origin as well as at infinity. Since the limit of the potentials $V_p$ and all derivatives of $V_p$ approach $0$ as $x$ approaches the origin, $V_p$ and all derivatives have a removable discontinuity at the origin (see Figure \ref{default}). So $V_p$ can be extended to a smooth function on the whole real line, while the Yukawa potentials cannot since they have an infinite discontinuity at the origin. In particular, the partial derivatives of the Yukawa potentials are not bounded. Also, the potentials $V_P$ have a singularity at infinity, as they are not integrable far from the origin, while the Yukawa potentials are integrable away from the origin in any $L^p$ space.
\begin{figure}
\begin{center}
\includegraphics[scale=0.7]{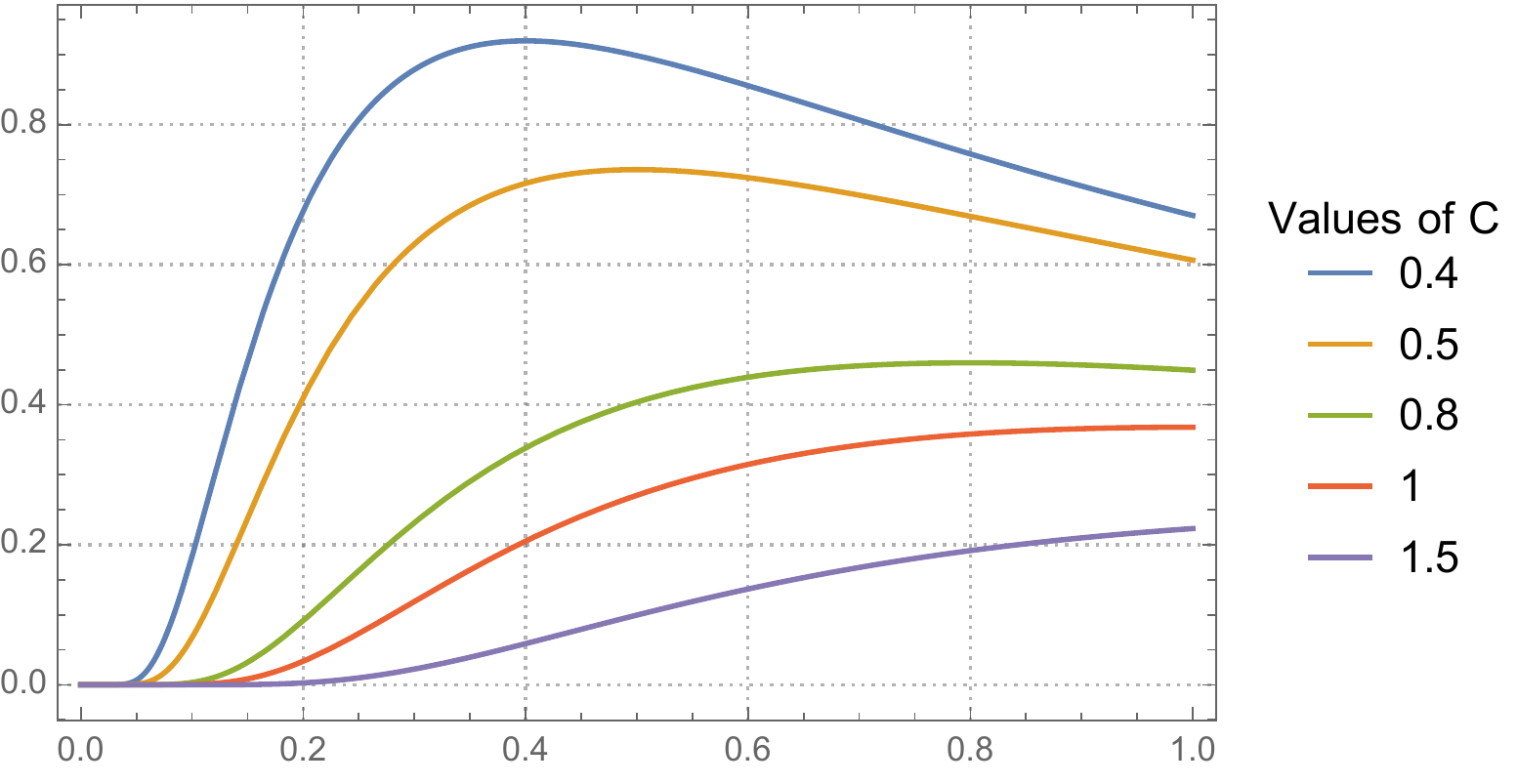}
\caption{Plots of the function $V_P(x)$ for various values of the constant $C$.}
\label{default}
\end{center}
\end{figure}
On the other hand, the Yukawa potentials were shown to be the potential for the strong force between nucleons, and have a very clear physical interpretation. A quick eigenvalue analysis reveals that $V_Y(x)$ are eigenfunctions of a general radial momentum operator
$$ P = -i \; \hbar \dfrac{1}{f(r)} \dfrac{\partial}{\partial r} \; f(r)$$
for appropriate $f(r)$, while the potentials $V_P(x)$ are eigenfunctions of the operator
$$ P_1 = r \dfrac{\partial}{\partial r} r$$
with eigenvalue $C$. We recall next a result which implies that the Yukawa potentials are not smooth enough to be studied in the context of Theorem \ref{thm:RG}. 
 
\begin{theorem}[Theorem 5, \cite{Lewin-preprint}] \label{thm:ML}
Let $V(r)$ be a radial, $C^\infty$ potential that at the origin satisfies
\begin{align} \label{derivativesatzero}
\dfrac{d^k V}{d r^k}(0) = 0 \qquad \forall \; k \geq 1.
\end{align}
Then the original Runge-Gross method applies. 
\end{theorem}
We should point out what it means exactly to say that this method applies. If (\ref{derivativesatzero}) holds, then the potential $V$ preserves the domain of $H_0^k$ for $k\geq 1$, i.e.
$$\mathscr{D}(H_V^k) = \mathscr{D}(H_0^k) \quad \forall \; k \geq 1$$
In general, since $H_V^k$ will be unbounded, $H_V^k$ needs to be understood in the sense of the functional calculus \cite{RSFA1}. The domain preservation means then that one can prove a Runge-Gross uniqueness theorem for the Hydrogen atom ($N=1$). Notice that the classical Coulomb potential does not satisfy (\ref{derivativesatzero}); hence the need for a smoothed out variant of the classical Coulomb potential. 

Now, we immediately see that the Runge-Gross method will not apply with Yukawa potentials\footnote{Without loss of generality we set $\alpha=1$.}. Indeed, we find
$$\dfrac{d V_Y}{d r} = -\dfrac{e^{-c r}(1+c r)}{r^2}$$
which does not vanish at the origin, since $c > 0$. 


However, by induction we have that (\ref{derivativesatzero}) holds for the potentials $V_P(x)$. Thus by Theorem \ref{thm:ML}, we have that
$$\mathscr{D}\left( H_{V_P}^k \right) = \mathscr{D} \left( H_0^k \right) \quad \forall \; k \geq 1$$
This in particular implies that if we take $V_0(x) = V_P(x)$ (rather than $V_0$ given by (\ref{V0singular})) and $V(x,t)$ a smooth external potential, then by \cite[Theorem 3]{Lewin-preprint}, a Runge-Gross uniqueness result holds in this setting. However, this is not necessarily the case if $V_0$ is given by (\ref{V0singular}). 

To this end, we allow for such singular fixed potentials $V_0$, and we introduce the time dependence into the potentials $V_P(x)$ by setting
\begin{align} \label{timedependenceVp}
V(x,t)= V_P(x-r(t)).
\end{align}
\noindent A similar setup of W\"uller \cite{Wuller86} gives existence of solutions to (\ref{Nparticletdse}) for certain singular time dependent potentials, provided the $N$ particle system is such that $N-1$ particles can be treated classically, and the last particle quantum mechanically. Our next section is devoted to solving the same equation with external potential given by (\ref{timedependenceVp}). 



%



\section{Existence of Solutions for $V_P$} \label{sec:existence}

Our main theorem for the Hydrogen atom is:


\begin{theorem} \label{maintheorem}
Consider the Hamiltonian (\ref{fullHamiltonian}) with $H_0$ given by (\ref{fullHamiltonianH0}) and the external potential given by
\begin{align*}
V(x,t) = V_P(x-r(t)) = \dfrac{e^{-C/|x-r(t)|}}{|x-r(t)|}
\end{align*}
with $r(t)$ continuous. Then, there exists a solution to the Schr\"odinger equation (\ref{Nparticletdse}) with $N=1$ given by a unitary propagator $U(t,s)$, such that the solution $\psi(t) = U(t,s)\psi_0$ is continuously differentiable in $L^2 (\mathbb{R}^3)$ for all $\psi_0 \in \mathscr{D}(H_0) \cap \mathscr{D}(x^2)$ and such that $\psi(t) \in \mathscr{D}(H_0) \cap \mathscr{D}(x^2)$ for $t>0$.
\end{theorem}

We mention here that the condition that $\psi_0 \in \mathscr{D}(x^2)$ is not unnatural; rather, most ground state wavefunctions in practice will satisfy this condition inherently. For instance, the ground state wavefunction of Hydrogen is
$$\psi_a(r) = a e^{-r}, \quad r = |x| = \sqrt{x^2+y^2 +z^2}$$
for some constant $a$. Then we clearly have $\psi_a \in L^2(\mathbb{R}^3)$ but we also have $r^2 \psi_a \in L^2(\mathbb{R}^3)$. This can be verified directly by computing the integral in spherical coordinates. 

Additionally, only continuity of the nuclear trajectory is required in Theorem \ref{maintheorem}. Hence the path does not need even to be differentiable.  

\subsection{Proof of Theorem \ref{maintheorem}}
We will use the following Theorem (which can be found in \cite{tanabe1979equations}) to prove Theorem \ref{maintheorem}. First note that $$\mathscr{D}(H_0) = \left\{ \psi \in L^2(\mathbb{R}^3): \Delta \psi \in L^2(\mathbb{R}^3)\right\}.$$
Note that in the previous example of the Hydrogen atom, it can be verified directly that $\psi_a \in \mathscr{D}(H_0)$.

\begin{theorem}[\cite{tanabe1979equations}] \label{Wullertheorem2}
Let $\{V(t)\}_{t \in [0, T]}$ be a family of bounded functions from $\mathbb{R}^3 \to \mathbb{R}$ such that:

\begin{enumerate}
\item $V(t) \mathscr{D}(H_0) \subset \mathscr{D}(H_0)$ for all $t$;
\item $V(t)$ is strongly continuous as a multiplication operator in $L^2(\mathbb{R}^3)$;
\item For all $\psi_0 \in Y$, the function $t \mapsto [H_0, V(t)]\psi_0$ is continuous in $L^2 (\mathbb{R}^3)$
where $Y = (\mathscr{D}(H_0 + x^2), ||\cdot ||_{Y})$ with $||\psi||_Y^2 = ||(H_0 + x^2)\psi||^2 + ||\psi||^2$. 
\end{enumerate}

Then, there exists a unitary propagator $U(t,s): Y \to Y$ such that $\psi_s(t) = U(t,s)\psi_0$ is continuously differentiable in $L^2(\mathbb{R}^3)$ for all $\psi_0 \in Y$ with
$$\dfrac{\partial U(t,s)\psi_0}{\partial t} =-iH_V(t) U(t,s) \psi_0.$$
\end{theorem}
The propagator from the above theorem is constructed using the so-called Dyson series \cite[Theorem X.69]{RS2}
$$U(t,s)\phi = 1+\sum_{j=1}^\infty (-\mathrm{i})^j \int_s^t \int_s^{t_1} \idotsint \int_s^{t_{j-1}} H(t_1) \cdots H(t_j) \phi dt_j ...dt_1$$

We proceed now to the 
\begin{proof}[Proof of Theorem \ref{maintheorem}]
We show that the family $V(t)$ defined by
\begin{align} \label{V2def}
\left(V(t)\right)[x] \coloneqq V_P(x-r(t))
\end{align}
satisfies the requirements of the previous theorem. For this it is enough to show that $V_P: \mathbb{R}^3 \to \mathbb{R}$ is bounded, $\partial_j V_p \in L^\infty (\mathbb{R}^3)$ for $1 \leq j \leq 3$, and 
$$\Delta V_p \in L^2 (\mathbb{R}^3) + L^\infty(\mathbb{R}^3)$$ in the sense of distributions\footnote{See the Lemma on pg. 169 in \cite{Wuller86}, or just verify directly.}. We check that the derivatives of $V_P(x)$ with respect to $x_j$ for $j=1,2,3$ are bounded. First,
		$$\dfrac{\partial}{\partial x_j}\dfrac{1}{|x|}=-\dfrac{x_j}{|x|^3}.$$
	Next, for all $C>0,$
		$$\dfrac{\partial V_p(x)}{\partial x_j}=-e^{-C/|x|}\dfrac{Cx_j}{|x|^4}-e^{-C/|x|}\dfrac{x_j}{|x|^3}.$$
	Both summands are in $ L^\infty (\mathbb{R}^3)$ because they are continuous, equal $0$ at the origin, and decay to $0$ as $|x|$ tends to $\infty.$ For this reason we also have that $V_P \in L^\infty (\mathbb{R}^3)$. 
	
	To compute the Laplacian, we simply take one more derivative to what was computed above:
		\begin{align*}
		& \dfrac{\partial}{\partial x_j} \dfrac{e^{-C/|x|}}{|x|}\Bigg(\dfrac{Cx_j}{|x|^3}+\dfrac{x_j}{|x|^2}\Bigg)\\
		& =e^{-C/|x|}\Bigg(\dfrac{Cx_j}{|x|^4}+\dfrac{x_j}{|x|^3}\Bigg)\Bigg(\dfrac{Cx_j}{|x|^3}+\dfrac{x_j}{|x|^2}\Bigg)+\dfrac{e^{-C/|x|}}{|x|}\Bigg(\dfrac{1}{|x|^2}+\dfrac{C}{|x|^3}-\dfrac{2x_j^2}{|x|^4}-\dfrac{3Cx_j^2}{|x|^5}\Bigg).
		\end{align*}
	After summing over $j,$ we see $\Delta V_p(x)$ is equal to 
		$$e^{-C/|x|}\Bigg(\dfrac{2}{|x|^3}+\dfrac{2C}{|x|^4}+\dfrac{C^2}{|x|^5}\Bigg).$$
	Notice that this can be written as a sum\footnote{Note that since $L^2(\mathbb{R}^3)$ is a Banach space, the sum of three $L^2(\mathbb{R})$ functions certainly belongs to $L^2 (\mathbb{R}^3)$.} of three functions in $L^2(\mathbb{R}^3),$ since $|x|^{\beta}$ is square integrable in $\mathbb{R}^3 \setminus B_1(0)$ for all $\beta<-3/2,$ and because $|x|^\beta e^{-C/|x|}$ is bounded on $B_1(0)$ for all $\beta\in\mathbb{R}.$

\end{proof}

\begin{remark}
We point out here the classical result of Kato in \cite{kato1953integration}, which provides another (less constructive) method to establish existence of solutions. This method relies on the Kato-Rellich Theorem, which in one form says that if two operators $A$ and $B$ are such that $A$ is self-adjoint and $B$ is symmetric and $A-$bounded with bound strictly less than one, then $A+B$ is self-adjoint on  $\mathscr{D}(A)$. Recall that $B$ is said to be $A$-bounded with bound $a$ if for all $\phi \in \mathscr{D}(A)$, there holds
$$||B\phi|| \leq a ||A\phi|| + b ||\phi||$$
for some $b \in \mathbb{R}$. We apply this theorem with the Hamiltonian $H_V (t) = H_0 +V(t)$ by taking $A=-\Delta$ and $B= V_0(x) - V(x,t)$ with $V_0$ given by (\ref{V0singular}) and $V(x,t)$ given by (\ref{timedependenceVp}). Since $V(x,t)$ is bounded, we conclude that $V_0(x) -V(x,t)$ is $-\Delta$-bounded with bound strictly smaller than one, and thus 
$$\mathscr{D}(H_V(t)) = \mathscr{D}\left( -\Delta - V_0(x) +V(x,t) \right) = \mathscr{D}(-\Delta)$$
is in particular independent of time. Then by \cite[Theorem 3]{kato1953integration}, a continuously differentiable solution of (\ref{Nparticletdse}) for $N=1$ can be found. However, it is unclear via this result if the solution would remain in the domain of the quantum harmonic oscillator for positive times. 
\end{remark}

\section{A general $N$ particle System} \label{sec:Nparticleexistence}
The setup of the $N$ particle, $M$ nuclei problem is as follows. We take
\begin{align} \label{NparticleVp}
V_P(x_j, t) \coloneqq - \sum_{m=1}^M \dfrac{Z_m e^{-C/|x_j -r_m(t)|}}{|x_j-r_m(t)|}
\end{align}
where $Z_m>0$ denotes the charge of nucleus $m$, and $x_j \in \mathbb{R}^3$ denotes the position of electron $j$, for $1 \leq j \leq N$. The underlying Hamiltonian takes the form
\begin{align} \label{NparticleHamiltonian}
H_V(t) = \underbrace{\sum_{j=1}^N -\Delta_{x_j} -\dfrac{1}{|x_j|}  +\sum_{1 \leq j < k \leq N} \dfrac{1}{|x_j - x_k|}}_{H_0} - \sum_{m=1}^M Z_m W(X -r_m(t)) 
\end{align}
where
$$W(X) = W(x_1,...,x_N) = \sum_{j=1}^N \dfrac{e^{-C/|x_j|}}{|x_j|}.$$
and the term$\sum_{j \neq k} |x_j-x_k|^{-1}$ is the Coulomb interaction between electrons. By Kato's Theorem \cite[Theorem X.16]{RS2}, we have that $H_0$ is essentially self-adjoint on $C_0^\infty (\mathbb{R}^{3N})$. That is, $H_0$ is densely defined, symmetric, and its closure $\overline H_0$ is self-adjoint. We denote the self-adjoint extension again by $H_0$.

In this case we consider the family of operators $\{V(t)\}_{t \in [0, T]}$ such that each $V(t): \mathbb{R}^{3N} \to \mathbb{R}$ are are given by
$$[V(t)](X) = [V(t)](x_1,...,x_N)= -\sum_{j=1}^N \sum_{m=1}^M \dfrac{Z_m e^{-C/|x_j -r_m(t)|}}{|x_j -r_m(t)|}.$$
We will show that the family $\{ V(t) \}$, as before, satisfies the requirements of Theorem \ref{Wullertheorem2}, now in $\mathbb{R}^{3N}$. It is a straightforward generalization of the proof of Theorem \ref{maintheorem} that items (1) and (2) hold in $\mathbb{R}^{3N}$. One has to be careful about item (3) in this theorem, however. Instead of considering $\mathscr{D}(H_0 + x^2) = \mathscr{D}(H_0) \cap \mathscr{D}(x^2)$, we have to consider now a general $N$ dimensional quantum harmonic oscillator Hamiltonian $H_0 + \sum_{j=1}^N x_j^2$. This has domain
$$\mathscr{D} \left( H_0 + \sum_{j=1}^N x_j^2 \right) = \mathscr{D}(H_0) \cap \mathscr{D} \left( \sum_{j=1}^N x_j^2\right) = \mathscr{D}(H_0) \cap \left[ \cap_{j=1}^N \mathscr{D}(x_j^2) \right].$$
This $N$ dimensional quantum harmonic oscillator Hamiltonian, a key tool used in the proof of Theorem \ref{Wullertheorem2}, is self-adjoint on $\mathscr{D}(H_0) = L^2(\mathbb{R}^{3N})$. Thus we have:

\begin{theorem} \label{Nparticletheorem}
Consider the Hamiltonian (\ref{NparticleHamiltonian}) as above with potential given by (\ref{NparticleVp}). Then there exists a solution to the $N$ particle Schr\"odinger equation (\ref{Nparticletdse}) given by a unitary propagator $U(t,s)$, such that the solution $\psi(t) = U(t,s)\psi_0$ is continuously differentiable in $L^2(\mathbb{R}^{3N})$ for all $\psi_0 \in \mathscr{D}(H_0) \cap  \left[ \cap_{j=1}^N \mathscr{D}(x_j^2) \right]$ and such that $\psi(t) \in \cap_{j=1}^N \mathscr{D}(x_j^2)$ for $t>0$.
\end{theorem}

We end this section by taking a specific ground state wavefunction as we did with the Hydrogen atom. Consider the instance the Helium atom ($N=2$), which has ground state wavefunction of the form
$$\psi_{He} (r_1, r_2) = k \cdot e^{-(r_1 + r_2)}, \qquad r_j = \sqrt{x_j^2 +y_j^2+ z_j^2},\;  \; j=1,2$$
for some constant $k$. Then the operator $H_0 = -\Delta_1 -\Delta_2$ where the $\Delta_j$ denotes the Laplacian acting on particle $j$. Then,
$$\mathscr{D}(H_0) = \mathscr{D}\left( -\Delta_1 -\Delta_2\right) = \mathscr{D}(-\Delta_1) \cap \mathscr{D}(-\Delta_2)$$
As in the case of the Hydrogen atom, we see that $\psi_{He} \in L^2 (\mathbb{R}^6)$ which we identify with $L^2(\mathbb{R}^3) \otimes L^2(\mathbb{R}^3)$. In particular, $\psi_{He} \in \mathscr{D}(H_0) \cap \mathscr{D}(r_1) \cap \mathscr{D}(r_2)$. 

\section{Fourier Analysis} \label{sec:FA}
In this section we give a rigorous analysis of the Fourier transform of $V_P$, and show that in the limit as the parameter $C$ approaches zero, we recover the Fourier transform of the Coulomb potential. This is frequently done with the potentials $V_Y$, but here we prove it can be done with $V_P$ as well. 

We view the potential $V_P$ not as a function but as a distribution. Since $V_P \not\in L^1 (\mathbb{R}^3)$, we take the Fourier transform as defined on the space of tempered distributions. We discuss this in detail next. 

\subsection{Fourier Transform and Distributions} \label{sec:Distributions}

Recall that a multi-index $\alpha$ in $\mathbb{R}^3$ is a vector $\alpha = (\alpha_1, \alpha_2, \alpha_3)$ where each $\alpha_j \in \mathbb{Z}_+$. The length of $\alpha$ is $|\alpha| = \alpha_1 + \alpha_2 + \alpha_3$. Then, given a multi-index $\alpha$, we can define
$$D^\alpha f \coloneqq \dfrac{\partial^{|\alpha|} f}{\partial x^{\alpha_1} \partial y^{\alpha_2} \partial z^{\alpha_3}} = \partial^{\alpha_1}_x \partial^{\alpha_2}_y \partial^{\alpha_3}_z f$$
for a differentiable function $f=f(x,y,z)$. Also, for a vector $x=(x_1,...,x_n) \in \mathbb{R}^n$ and a multi-index $\alpha$, we define $x^\alpha = x_1^{\alpha_1} x_2^{\alpha_2}\cdots x_n^{\alpha_n}$. Furthermore, the support of a function $f$ is the closure of the set of points where $f$ does not vanish. 

We will define the Fourier transform on a subset of distributions, and for this we need the following definitions.

\begin{definition}
Let $\Omega \subset \mathbb{R}^3$ be an open set. A linear functional $u: C_0^\infty (\Omega) \to \mathbb{C}$ is called a distribution if for every compact set $K \subset\Omega$, there exists $C \geq 0$ and a positive integer $N$ such that
$$\left| \int u \phi dx \right| \leq C \sum_{|\alpha| \leq N} \sup |\partial^\alpha \phi|$$
for all $\phi \in C_0^\infty (\Omega)$ whose support lies in $K$. The space of distributions on $\Omega$ is denoted $\mathscr{D}'(\Omega)$.
\end{definition}

\begin{definition}
A function $\phi \in C^\infty (\mathbb{R}^3)$ is called rapidly decreasing if 
$$\sup_{x \in \mathbb{R}^3} |x^\alpha D^\beta \phi| < \infty$$
for all multi-indices $\alpha, \beta$. The space of such functions is denoted $\mathscr{S}(\mathbb{R}^3)$. 
\end{definition}

In particular, a sequence $\{\phi_j\}_j \in \mathscr{S}(\mathbb{R}^3)$ is said to converge to zero in $\mathscr{S}(\mathbb{R}^3)$ if 
$$||\phi_j||_{\alpha, \beta} \coloneqq \sup_x |x^\alpha D^\beta \phi| \longrightarrow 0$$ as $j \to \infty$ for all $\alpha, \beta$. Now, the dual of $\mathscr{S}(\mathbb{R}^3)$ is the collection of all continuous linear functionals on $\mathscr{S}(\mathbb{R}^3)$. This space is denoted $\mathscr{S}'(\mathbb{R}^3)$, and it is on this space we define the Fourier transform (the members of this space are called tempered distributions). 

If $u \in L^1 (\mathbb{R}^3)$ then its classical Fourier transform $\widehat u(\xi)$ can be viewed as an element of $\mathscr{S}'(\mathbb{R}^3)$ (see \cite[Theorem 8.1.3]{Friedlander98}). Using this we can define the Fourier transform on $\mathscr{S}'(\mathbb{R}^3)$ as follows. 

\begin{definition}
The Fourier transform of $u \in \mathscr{S}'(\mathbb{R}^3)$ is the distribution $\widehat u \in \mathscr{S}'(\mathbb{R}^3)$ defined by
$$\int \widehat u(\xi) \phi(\xi) d\xi = \int u(y) \widehat \phi(y) dy, \qquad \phi \in \mathscr{S}(\mathbb{R}^3)$$
\end{definition}
The fact that the Fourier transform makes sense on $\mathscr{S}(\mathbb{R}^3)$ follows from \cite[Theorem 8.2.3]{Friedlander98}. Finally, if $u \in L^1(\mathbb{R}^3)$, then its distributional Fourier transform agrees with its ``classical" Fourier transform.

\subsection{Computing the Fourier Transform}
We can compute explicitly the Fourier transform of $V_P$, but in taking the limit as $C\to 0$, we need to understand this in the distributional sense. 

To this end we employ the radial Fourier transform, which can be defined using the Bessel function $J_\nu(t)$, defined by

$$J_\nu (t) \coloneqq \dfrac{t^{\nu}}{(2\pi)^{\nu +1}} \omega_{2\nu} \int_0^\pi e^{-i t \cos(\theta)} \sin(\theta)^{2\nu} d\theta, \qquad \nu \in \mathbb{R}$$
where $\omega_n$ denotes the area of the unit sphere $S^n$. Now, the Fourier transform of a radial function $F= F(r)$ in $\mathbb{R}^3$ is given by \cite{SteinFourier}

\begin{align} \label{FTofRadial}
\widehat F (\xi) = 2 \int_0^\infty \dfrac{\sin(2\pi r|\xi|)}{|\xi|} r F(r)dr
\end{align}
This representation relies on the fact that
$$J_{1/2}(t) = \dfrac{t^{1/2}}{\sqrt{2\pi}} \cdot \dfrac{2 \sin(t)}{t}$$
Now, since $V_P \not\in L^1(\mathbb{R}^3)$, to employ the classical Fourier transform we introduce the function
$$V_P^k (\xi) \coloneqq \dfrac{e^{-C/|\xi|}}{|\xi|} e^{k|\xi|}, \qquad k < 0$$
Then the function $V_P^k \in L^1(\mathbb{R}^3)$ so we can take its Fourier transform using (\ref{FTofRadial}). Indeed, 

$$
\widehat V_P^k (\xi) = \dfrac{2}{|\xi|} \int_0^\infty \sin(2\pi r |\xi|) e^{-C/r} e^{k r} dr
$$
Next we write $\sin(2\pi r |\xi|) = {\rm Im}  \left( e^{2\pi i r |\xi|} \right)$ to get
$$\widehat V_P^k (\xi) = {\rm Im}  \left(  \dfrac{2}{|\xi|} \int_0^\infty e^{2\pi i r |\xi|} e^{-C/r} e^{k r} dr \right)$$
To calculate the integral above, we appeal to \cite[Equation (2.11)]{SpecialFunctions}, which says
\begin{align} \label{specialfunctionidentity1}
2 \left( \dfrac{a}{b} \right)^{1/2} K_1 (2\sqrt{ab}) = \int_0^\infty \exp \left( -a/ r - br \right) dr
\end{align}
where $K_1 (z)$ denotes the modified Bessel function of the first kind, given by \cite{AS64}
$$K_1(z) = \lim_{\alpha \to 1} \frac{\pi}{2} \dfrac{I_{-\alpha}(z)-I_{\alpha}(z)}{\sin(\alpha \pi)}$$
where
$$I_{\nu}(z) = \sum_{k=0}^\infty \dfrac{1}{\Gamma \left(\nu + 1 + k\right)} \left( \dfrac{z}{2} \right)^{\nu+2k}$$ 
So, using (\ref{specialfunctionidentity1}) with $a=C$ and $b=-2\pi i |\xi| - k$ we obtain
\begin{align} \label{FTofVPk}
\widehat V_P^k(\xi) = \dfrac{2}{|\xi|} {\rm Im}  \left( \dfrac{2 K_1 \left( \sqrt{-C(k+2\pi i |\xi|)}\right)}{\sqrt{ -\dfrac{k+2\pi i |\xi|}{C}}}\right)
\end{align}
By continuity we then have
$$\lim_{k \to 0} V_P^k (\xi) = \dfrac{2}{|\xi|} {\rm Im} \left(\dfrac{ 2K_1 \left( \sqrt{-2\pi i C |\xi|}\right)}{\sqrt{\dfrac{-2\pi i |\xi|}{C}}}\right)$$
Since $\lim_{k \to 0} \widehat V_P^k (\xi) = \widehat V_P(\xi)$ in the sense of distributions, we conclude that the Fourier transform of $V_P$ (as a tempered distribution) is
$$\widehat V_P (\xi) = \dfrac{2}{|\xi|} {\rm Im} \left(\dfrac{ 2K_1 \left( \sqrt{-2\pi i C |\xi|}\right)}{\sqrt{\dfrac{-2\pi i |\xi|}{C}}}\right)$$
Finally, we claim that
\begin{align} \label{limitclaim}
\lim_{C\to 0} \widehat V_P (\xi) = \dfrac{2}{|\xi|} \dfrac{i}{2\pi |\xi|}
\end{align}
Assuming the claim, using that $\widehat V_P$ converges to the Fourier transform of the Coulomb potential in the sense of distributions, we recover the expected Fourier transform of the Coulomb potential. So it remains to show (\ref{limitclaim}). From \cite[Equation 9.6.9]{AS64}, we have that $K_1 (w) \sim w^{-1}$ as $w\to 0$, which means that for $C \to 0$ we have
$$K_1 \left( \sqrt{-2\pi i C |\xi|} \right) \sim \dfrac{1}{\sqrt{-2\pi i C |\xi|}}$$
This in particular implies that
$$\widehat V_P(\xi) \to \dfrac{2}{|\xi|} {\rm Im} \dfrac{C_0 i}{|\xi|}$$
for some constant $C_0$. Then taking the imaginary part gives the desired result.

\section{Spectral Properties of the Hamiltonian with Potential $V_P$} \label{sec:spectral}
To characterize the spectrum of the Hamiltonian with potential $V_P$, it is useful to consider the dilatation transformation $r \mapsto e^{i\theta}r$ for $\theta > 0$. This map rotates the coordinates of the Hamiltonian into the complex $r$ plane.  This is useful because the spectra of such Hamiltonians is well understood. 

Following \cite{Simon73}, we define the following class of operators on $L^2 (\mathbb{R}^3)$. First, denote by $u(\theta)$ the one parameter family of dilatations on $L^2(\mathbb{R}^3)$ given by
$$(u(\theta)f)(|x|)\coloneqq e^{3\theta/2} f(e^\theta |x|), \qquad \theta \in \mathbb{R}$$
The factor $e^{3\theta/2}$ is introduced to make the operator unitary. Notice that the kinetic energy $-\Delta$ transforms very nicely under $u$:
$$u(\theta)\left( - \Delta \right) u(\theta)^{-1} = e^{-2\theta} (-\Delta)$$
which implies that $u(\theta) (-\Delta) u(\theta)^{-1}$ has analytic continuation to complex $\theta$. The idea behind dilatation analytic potentials is to allow this result to remain true if a potential is added to the kinetic energy operator. 

Further, let $\mathscr{H}$ denote the Hilbert space $L^2 (\mathbb{R}^6, d^6x)$. 

\begin{definition}
Let $\alpha \in \mathbb{R}_+$. We say an unbounded operator $V$, defined on $L^2(\mathbb{R}^3)$, belongs to the class $C_\alpha$, if:
\begin{enumerate}
\item $\mathscr{D}(V) = \mathscr{D}(-\Delta)$ on $L^2(\mathbb{R}^3)$, and $V$ is symmetric; 
\item The induced operator $V: H^2 \coloneqq \mathscr{D}(-\Delta) \to \mathscr{H}$ is bounded and compact; 
\item The operators $V(\theta) : H^2 \to \mathscr{H}$ defined by $V(\theta) = u(\theta) V u(\theta)^{-1}$ for $\theta \in \mathbb{R}$ have an analytic continuation to the strip $\{ \theta: | \text{Im}(\theta) | < \alpha \}$. 
\end{enumerate}
where $H^k$ denotes the standard Sobolev space of order $k$, for $k \in \mathbb{Z}$, on $\mathbb{R}^3$. 
\end{definition}
The point of introducing this class of potentials is that the spectrum of such operators is explicitly known \cite{Balslev-Combes1}. We state the result of Balslev and Combes for a general $N$ electron system. To this end let $U(\theta)$ be the group of dilatations on $\mathbb{R}^{3(N-1)}$ given by
$$(U(\theta)f)(x) \coloneqq e^{3(N-1)\theta/2} f(e^{\theta} x)$$
For the $N$ electron Hamiltonian $H$ define $H(\theta) \coloneqq U(\theta) H U(\theta)^{-1}$. Write the potential $V = \sum_{j<k} V_{jk}(|x_j-x_k|)$. Denote by $D =\{ D_1, ..., D_k \}$ a \emph{decomposition} of $\{0,1,...,N-1\}$ into $k\geq 2$ groups. That is, $D_i\cap D_j = \emptyset$ if $i\neq j$ and
$$\bigcup_{i=1}^k D_i = \{0, 1,..., N-1\}$$
Let $H_i$ denote the Hamiltonian for the group $D_i$, i.e $H_i = -\Delta_i + V_i$, where $V_i$ is the set of interactions between electrons in group $D_i$. A bound state energy $E_1+...+E_k$ with $E_i$ eigen-energy of $H_i$ will be called a \emph{k-body threshold}. The family of all of these will be denoted $\Sigma$; a similar analysis can be done with $H(\theta)$. Denote the family of thresholds of $H(\theta)$ by $\Sigma(\theta)$. The result of Balslev and Combes then says:

\begin{theorem}[\cite{Balslev-Combes1}]
Assuming each $V_{jk} \in C_\alpha$, the spectrum of $H(\theta)$ is given explicitly as:
\begin{enumerate}
\item $\{z +e^{-2\theta} x: x \in \mathbb{R}_+, \; \forall \; z \in \Sigma(\theta) \}$; 
\item A set of isolated points $\sigma_i$ of the spectrum, which are eigenvalues with finite multiplicity. 
\end{enumerate}
Additionally, the real eigenvalues and thresholds of $H(\theta)$ are exactly those of $H$. All non-real eigenvalues and thresholds of $H(\theta)$ lie in the sector
$$\{ z: 0> \arg (z- \Sigma_{\text{min}}) < -2 \; \text{Im}(\theta) \}$$
where $\Sigma_{\text{min}} = \inf \{y: y \in \Sigma \cap \mathbb{R} \}$. 

Finally, the complex eigenvalues and thresholds of $H(\theta)$ that are isolated from other parts of the essential spectrum of $H(\theta)$ are in the spectrum of $H(\theta')$ if the imaginary part of $\theta'$ is close enough to the imaginary part of $\theta$. 

\end{theorem}
In particular, the class $C_\alpha$ is large enough to contain the classical Coulomb potential $1/r$ \emph{and} the Yukawa potential. In fact, it was shown in \cite{Simon73} that these potentials actually belong to form analogs of $C_\alpha$, denoted $\mathscr{F}_\alpha$. In particular, $r^{-1} \in \mathscr{F}_{\infty}$ and $V_Y \in \mathscr{F}_{\pi/2}$. However, by replacing perturbation theory of ``Type A" with perturbation theory of ``Type B" (in the sense of Kato \cite{kato2013perturbation}), all results for $C_\alpha$ hold for $\mathscr{F}_\alpha$. 

To show that $V_P$ are indeed dilatation analytic, we appeal to the following useful characterization of such potentials. 

\begin{theorem}[\cite{Babbitt-Balslev1975}]
Let $\alpha \in \mathbb{R}_+$ and denote $S_\alpha = \left\{ re^{i \phi}: 0 < r < \infty, \; -\alpha < \phi < \alpha\right\}$. If $v(\xi)$ is an analytic function on $S_\alpha$ satisfying
\begin{enumerate}
\item $\sup\limits_{-\alpha + \epsilon \leq \phi \leq \alpha -\epsilon} \int_0^1 |v(re^{i \phi})|^2 r^2 dr < \infty$\\
\item $\sup\limits_{-\alpha + \epsilon \leq \phi \leq \alpha -\epsilon} \sup\limits_{1 \leq r \leq \infty} |v(re^{i \phi})| < \infty$
\end{enumerate}
for every $\epsilon > 0$, then the multiplication operator $V$ corresponding to the function $v$ is dilatation analytic in $S_\alpha$. 
\end{theorem}
From this it was shown in \cite{simon1978resonances} that for a radial potential it's actually enough to check the following:

\begin{enumerate}
\item[(I)] $V(r)$ has an analytic extension to the sector $|\arg (r)| < \alpha$
\item[(II)] $$\lim_{\substack{r \to \infty \\ |\arg(r)| < \beta}} |V(r)| =0 \qquad \forall \; \beta < \alpha$$ 
\item[(III)] $$\lim_{\substack{r\to 0 \\ |\arg(r)| < \beta}} r^{2-\epsilon} |V(r)| = 0 \quad \text{ for some } \epsilon > 0 \; \text{ and all } \beta < \alpha$$
\end{enumerate}
Item (I) holds with $\alpha = \pi/2$ by a classical result of Whitney \cite{whitney1934analytic}, and the other two items follow easily by definition of $V_P(r)$ (one can take $\epsilon = 1/2$ for (III)).

\begin{remark}
One can do a much more in depth eigenvalue analysis of Hamiltonians with $V_P(x)$ than what we've done here. In particular, since the spectrum and eigenvalues for Coulombic Hamiltonians are well known, one can studying the limiting case (as the parameter $C\to 0$) of the eigenvalues of our Hamiltonians to Coulombic ones. This could be useful for example in doing resonance calculations for the Helium atom, for which one expects to observe some sort of ``anomalies" in Auger processes \cite{Simon73}. The width of such an anomaly, called a resonance, can be calculated in terms of the eigenvalues of the Hamiltonian. In particular, viewing $V_P$ as a perturbation of $H_0$, it would be interesting to see if an eigenvalue can be turned into a resonance; we have not yet pursued this direction. 
\end{remark}


\section{Conclusion and Outlook}
We have introduced a new class of softened Coulomb potentials $V_P(x)$ and shown existence of solutions to a time dependent Schr\"odinger equation for a general $N$ particle, $M$ nuclei problem with such potentials. We also found explicitly the Fourier transform of these potentials, and used this to give an alternative rigorous derivation of the Fourier transform of the Coulomb potential. Finally we showed that the potentials are dilatation analytic, and so the spectrum of the Hamiltonian can be understood completely. The next step is to use the potentials to aid in the development of a more general Runge-Gross Theorem. In particular, by showing that a Runge-Gross type theorem holds for $V_P$, one can study the limit (in an appropriate sense) as the parameter $C \to 0$ to recover the Coulomb potential in this context. This will be done in future work. A more thorough eigenvalues analysis of Hamiltonians with perturbations $V_P$ will also be done, as discussed in the previous section.

\bibliographystyle{abbrv}
\bibliography{TDDFTbib}

\begin{thebibliography}{10}

\bibitem{AS64}
M.~Abramowitz and I.~Stegun.
\newblock {\em Handbook of mathematical functions: with formulas, graphs, and
  mathematical tables}.
\newblock Courier Corporation, 1964.

\bibitem{Babbitt-Balslev1975}
D.~Babbitt and E.~Balslev.
\newblock A characterization of dilatation-analytic potentials and vectors.
\newblock {\em Journal of Functional Analysis}, 18:1--14, 1975.

\bibitem{Balslev-Combes1}
E.~Balslev and J.~M. Combes.
\newblock Spectral properties of many-body schr{\"o}dinger operators with
  dilatation-analytic interactions.
\newblock {\em Communications in Mathematical Physics}, 22(4):280--294, 1971.

\bibitem{Lewin-preprint}
S.~Fournais, J.~Lampart, M.~Lewin, and T.~O. Sorensen.
\newblock Coulomb potentials and taylor expansions in time-dependent density
  functional theory.
\newblock {\em Physical Review A}, 93(6), 2016.

\bibitem{Friedlander98}
F.~G. Friedlander and M.~S. Joshi.
\newblock {\em Introduction to the Theory of Distributions}.
\newblock Cambridge University Press, 1998.

\bibitem{SpecialFunctions}
L.~Glasser, K.~T. Kohl, C.~Koutschan, V.~H. Moll, and A.~Straub.
\newblock {The integrals in Gradshteyn and Ryzhik. Part 22: Bessel-K
  functions}.
\newblock {\em Scientia, Series A, Mathematical Sciences}, 22:129--151, 2012.

\bibitem{kato1953integration}
T.~Kato.
\newblock Integration of the equation of evolution in a banach space.
\newblock {\em Journal of the Mathematical Society of Japan}, 5(2):208--234,
  1953.

\bibitem{kato2013perturbation}
T.~Kato.
\newblock {\em Perturbation theory for linear operators}, volume 132.
\newblock Springer Science \& Business Media, 2013.

\bibitem{maitra2016perspective}
N.~T. Maitra.
\newblock Perspective: Fundamental aspects of time-dependent density functional
  theory.
\newblock {\em The Journal of Chemical Physics}, 144(22):220901, 2016.

\bibitem{marques2006time}
M.~Marques.
\newblock {\em Time-dependent density functional theory}, volume 706.
\newblock Springer Science \&amp; Business Media, 2006.

\bibitem{RS2}
M.~Reed and B.~Simon.
\newblock {\em Methods of Modern Mathematical Physics, Vol. 2: Fourier
  Analysis, Self-Adjointness}.
\newblock Academic press New York, 1975.

\bibitem{RSFA1}
M.~Reed and B.~Simon.
\newblock {\em Methods of Modern Mathematical Physics I: Functional Analysis}.
\newblock Academic Press, 1980.

\bibitem{Runge-Gross}
E.~Runge and E.~K.~U. Gross.
\newblock Density functional theory for time-dependent systems.
\newblock {\em Physical Review Letters}, 52(12):997--1000, 1984.

\bibitem{Simon73}
B.~Simon.
\newblock Resonances in n-body quantum systems with dilatation analytic
  potentials and the foundations of time-dependent perturbation theory.
\newblock {\em Annals of Mathematics}, pages 247--274, 1973.

\bibitem{simon1978resonances}
B.~Simon.
\newblock Resonances and complex scaling: a rigorous overview.
\newblock {\em International Journal of Quantum Chemistry}, 14(4):529--542,
  1978.

\bibitem{SteinFourier}
E.~M. Stein and G.~Weiss.
\newblock {\em Introduction to Fourier analysis on Euclidean spaces (PMS-32)},
  volume~32.
\newblock Princeton university press, 2016.

\bibitem{tanabe1979equations}
H.~Tanabe.
\newblock {\em Equations of evolution}, volume~6.
\newblock Pitman Publishing, 1979.

\bibitem{whitney1934analytic}
H.~Whitney.
\newblock Analytic extensions of differentiable functions defined in closed
  sets.
\newblock {\em Transactions of the American Mathematical Society},
  36(1):63--89, 1934.

\bibitem{Wuller86}
U.~W\"uller.
\newblock Existence of the time evolution for schr{\"o}dinger operators with
  time dependent singular potentials.
\newblock {\em Annales de l'IHP Physique th{\'e}orique}, 44(2):155--171, 1986.

\bibitem{Yajima2016}
K.~Yajima.
\newblock Existence and regularity of propagators for multi-particle
  schr\"odinger equations in external fields.
\newblock {\em Communications in Mathematical Physics}, 347(1):103--126, 2016.

\bibitem{Yukawa35}
H.~Yukawa.
\newblock On the interaction of elementary particles i.
\newblock {\em Proceedings of the Physico-Mathematical Society of Japan},
  17:48--57, 1935.

\end{thebibliography}

\end{document}